\documentclass[a4paper,UKenglish,cleveref, autoref, thm-restate]{lipics-v2021}

\usepackage{amsmath}
\usepackage[linesnumbered]{algorithm2e}

\usepackage[textsize=footnotesize]{todonotes}
\usepackage[utf8]{inputenc}

\usepackage{enumerate}
\usepackage{float}
\usepackage{graphicx}

\graphicspath{ {figures/} }

\usepackage{algorithm2e}

\usepackage{xcolor}

\newcommand{\sep}[1]{\textsf{sep}(#1)}
\newcommand{\psep}[1]{\textsf{psep}}

\newcommand{\weight}[1]{\textsf{weight}(#1)}

\newcommand{\old}[1]{{}}

\newcommand{\rect}[2]{{#1}_{R}}

\usepackage{complexity}
 \nolinenumbers

\title{Space-Efficient Algorithms for Reachability in Directed Geometric Graphs}

\author{Sujoy Bhore}{Indian Institute of Science Education and Research, Bhopal, India.}{sujoy.bhore@gmail.com}{0000-0003-0104-1659}{}

\author{Rahul Jain}{Fernuniversit\"at in Hagen, Germany}{rahul.jain@fernuni-hagen.de}{https://orcid.org/0000-0002-8567-9475}{}
\authorrunning{S. Bhore and R. Jain}

\Copyright{Sujoy Bhore and Rahul Jain}
\hideLIPIcs


\category{}

\relatedversion{}

\EventEditors{John Q. Open and Joan R. Access}
\EventNoEds{2}
\EventLongTitle{42nd Conference on Very Important Topics (CVIT 2016)}
\EventShortTitle{CVIT 2016}
\EventAcronym{CVIT}
\EventYear{2016}
\EventDate{December 24--27, 2016}
\EventLocation{Little Whinging, United Kingdom}
\EventLogo{}
\SeriesVolume{42}
\ArticleNo{23}

\begin{document}

\maketitle

\begin{abstract}
The problem of graph \textsc{Reachability} is to decide whether there is a path from one vertex to another in a given graph. In this paper, we study the \textsc{Reachability} problem on three distinct graph families - intersection graphs of Jordan regions, unit contact disk graphs (penny graphs), and chordal graphs. For each of these graph families, we present space-efficient algorithms for the \textsc{Reachability} problem.

For intersection graphs of Jordan regions, we show how to obtain a ``good'' \emph{vertex separator} in a space-efficient manner and use it to solve the \textsc{Reachability} in
polynomial time and $O(m^{1/2}\log n)$ space, where $n$ is the number of Jordan regions, and $m$ is the total number of crossings among the regions.
We use a similar approach for chordal graphs and obtain a polynomial-time and
$O(m^{1/2}\log n)$ space algorithm, where $n$ and $m$ are the number of vertices and edges, respectively. However, we use a more involved technique for unit contact disk graphs (penny graphs) and obtain a better algorithm.
We show that for every $\epsilon> 0$, there exists a polynomial-time algorithm that can solve \textsc{Reachability} in an $n$ vertex directed penny graph, using $O(n^{1/4+\epsilon})$ space. We note that the method used to solve penny graphs does not extend naturally to the class of geometric intersection graphs that include arbitrary size cliques.
\end{abstract}

\section{Introduction}
Given a directed graph $G=(V,E)$ and two of its vertices $s$ and $t$, the problem of \textsc{Reachability} is to decide if there exists a path from $s$ to $t$ in $G$. \textsc{Reachability} is one of the fundamental problems in theoretical computer science, and
for directed and undirected graphs the problem is known to be complete for the
classes \textsf{NL} and \textsf{L}, respectively (see~\cite{lewis1982symmetric, reingold2008undirected}).
The famous open question {\L} $\stackrel{?}{=}$ \NL~essentially asks if there exists
a deterministic log-space algorithm for \textsc{Reachability} or not.
Note that \textsc{Reachability} can be solved in $\Theta(n\log n)$ space and optimal time by using standard graph traversal algorithms such as DFS and BFS. Furthermore, it is known that this problem can be solved in $\Theta(\log^2 n)$-space and $n^{\Theta(\log n)}$ time~\cite{savitch1970relationships}.

In the realm of space-efficient algorithms,
the primary objective is to optimize the space-complexity of an algorithm while maintaining a polynomial-time bound. Wigderson asked in his survey of \textsc{Reachability} problems~\cite{wigderson1992complexity}, that whether there is an algorithm for \textsc{Reachability} that runs in $O(n^{1-\varepsilon})$ space (for any $\varepsilon>0$) and polynomial time.
Barnes et al.~\cite{barnes1998sublinear}
partially answered this question and showed that \textsc{Reachability} on general graphs can be solved in polynomial time and $O(n/2^{\Theta(\sqrt{\log n})})$ space. This result is followed by numerous works on various restricted graph families.
Asano and Doerr \cite{asano2011memory} presented an algorithm for grid graphs that uses $O(n^{1/2+\varepsilon})$-space, for any small $\varepsilon>0$. Imai et al.~\cite{DBLP:conf/coco/ImaiNPVW13} achieved the similar space bound for planar graphs. Later, Asano et al.~\cite{asano2014widetilde} improved the space bound to $\Tilde{O}(n^{1/2})$ for planar graphs. Recently, this bound has been improved to $O(n^{1/4+\varepsilon})$ for grid graphs \cite{DBLP:conf/fsttcs/JainT19}. Besides, Chakraborty et al.~\cite{DBLP:conf/fsttcs/ChakrabortyPTVY14} studied
\textsc{Reachability} for graphs with higher genus and gave an $\Tilde{O}(n^{2/3}g^{1/3})$-space bound algorithm, and an $\Tilde{O}(n^{2/3})$ space algorithm for \textsf{H} minor-free graphs.
For layered planar graphs, Chakraborty and Tewari \cite{DBLP:journals/toct/ChakrabortyT18} showed that, there is an $O(n^{\varepsilon})$-space and polynomial algorithm. Gupta et al.~\cite{gupta2019reachability} showed that given a pair $\{s,t\}$ and an embedding of an $O(\log n)$ genus graphs, \textsc{Reachability} from $s$ to $t$ in $G$ can be decided \emph{unambiguously} in logspace (i.e. \textsc{Reachability} is in the class \UL).

Tree-decomposition and the associated \emph{treewidth} is an essential notion of the graphs. Many problems which are computationally hard on general graphs are efficiently solvable on graphs of bounded Treewidth~\cite{arnborg1989linear}. Graphs of small Treewidth also have small \emph{vertex-separators}, which is a small set of vertices of the graph, removal of which divides the graph into pieces whose size are at most a fraction of the original graph. Recently, Jain and Tewari \cite{DBLP:conf/isaac/JainT19} showed that given an $n$ vertex directed graph of treewidth $w$ along with its tree decomposition, there exists an algorithm for \textsc{Rechability} problem that runs polynomial time and $O(w\log n)$ space. They achieved this result by using the vertex separator of small Treewidth graphs. They formalized the connection between \emph{Vertex Separator} and \textsc{Reachability} problems, which has several consequences. For the sake of completeness, we state their result below.

\begin{theorem}[\cite{DBLP:conf/isaac/JainT19}]\label{treewidth}
Let $\mathcal{G}$ be a class of directed graphs and $w:\mathbb{N}\times \mathbb{N} \rightarrow \mathbb{N}$ be a function. If there exist an $O(w(n,m)\log n)$ space and polynomial time algorithm, that given a graph $G\in \mathcal{G}$ of $n$ vertices and $m$ edges, and a set $U$ of $V(G)$, outputs a separator of $U$ in underlying undirected graph of $G$\footnote{For a directed graph $G$, its \emph{underlying undirected graph} is the graph formed by removing the orientation of all the edges.} of size $O(w(n,m))$, then there exists an algorithm to decide \textsc{Reachability}
in $G$ that uses $O(w(n,m)\log n)$ space and polynomial time.
\end{theorem}

In this paper, we study the \textsc{Reachability} problem on directed geometric graph families.
Many important graph families can be described as intersection graphs of more restricted types of set families. More often than not, the geometric intersection graph families provide additional geometric structures that help to generate efficient algorithms. Many problems that are \textsf{NP}-complete~on general graphs are tractable on geometric intersection graphs~\cite{clark1991unit, imai1983finding}. Indeed, such advantages have been exploited for space-efficient algorithms as well. We refer to the survey on geometric algorithms with limited work-space \cite{DBLP:journals/corr/abs-1806-05868}.

In this work, we make progress towards answering the question raised by Wigderson in his
survey~\cite{wigderson1992complexity}. We primarily design and use the \emph{Vertex separator} of some geometric graphs space-efficiently. Previously, vertex separators in the context of geometric graphs have been studied by Fox and Pach~\cite{fox2008separator}, who established several geometric extensions of the famous Lipton-Tarjan separator theorem~cite. Recently, Carmi et al.~\cite{DBLP:journals/comgeo/CarmiCKKORRSS20}, and Hoffmann et al.~\cite{hoffmann2014halving} established improved bounds on the size of separators for some restricted classes of geometric graphs.

\subsection{Our Contribution}
We study the \textsc{Reachability} problem on three different graph families - intersection Graphs of Jordan Regions,
Unit Contact Disk Graphs (Penny Graphs) and Chordal Graphs.

First, in Section~\ref{jordon}, we show that given a collection of Jordan regions, there exists a polynomial-time algorithm that computes a separator of size $O(m^{1/2})$ using $O(m^{1/2}\log n)$ space. Then, by combining this with Theorem~\ref{treewidth}, we note that \textsc{Reachability} on directed intersection graphs of Jordan regions can be solved in polynomial time and $O(m^{1/2}\log n)$ space, where $n$ is the number of Jordan regions, and $m$ is the total number of crossings among the regions.

In Section~\ref{penny}, we present a space-efficient algorithm for \textsc{Reachability} on penny graphs that uses $O(n^{1/4 + \varepsilon})$ space and polynomial time. Since penny graphs are a subclass of planar graphs, \textsc{Reachability} can be solved for penny graphs in $O(n^{1/2}\log n)$ space.
However, to reduce the space complexity, we use an involved technique.
First, we use the axis-parallel separator of Carmi et al.~\cite{DBLP:journals/comgeo/CarmiCKKORRSS20} repeatedly to form rectangular subdivision such that each cell of a subdivision contains a bounded size subgraph of the input graph. Then using these subdivisions, we construct a smaller auxiliary graph that preserves \textsc{Reachability} information. Finally, using a notion of \emph{pseduo-separator}, we solve \textsc{Reachability} in this auxiliary graph in a space-efficient manner.
Note that, there exists a a $O(n^{1/4 + \varepsilon})$ space and polynomial-time algorithm for reachability in grid graphs~\cite{DBLP:conf/fsttcs/JainT19}. Since grid graphs are a subclass of penny graphs, our result is a generalization and improvement of previously known results.

Finally, in Section~\ref{chordal}, we adopt the algorithm of~\cite{gilbert1984separator} for \textsc{Reachability} on chordal graphs to provide a space-efficient and polynomial-time algorithm.

Even though chordal graphs seem to be an exception in the context of geometric graph families, we note that there exists work representing chordal graphs as a subfamily of string graphs~\cite{chaplick2012bend}.
Moreover, it is known that chordal graphs can be characterized as intersection graphs of sub-trees of a tree~\cite{gavril1974intersection},
Also, Interval graphs are a subfamily of chordal graphs.

\subsection{Preliminaries.}
Throughout the text, we denote the set $\{1,2,\ldots,n\}$ as $[n]$. For a graph $G=(V,E)$ and a subset $U \subseteq V$, $G[U]$ denotes the subgraph induced on $U$.
Given a graph $G=(V,E)$ with vertex set $V$ and edge set $E$, a weight function $w:V\rightarrow R_{\ge 0}$ is a non-negative function on $V$ such that the sum of the weights is $1$.
For any subset $S\subset V$, the weight $w(S)$ is defined to be $\sum_{v\in S}w(v)$. A separator in a graph $G$ with respect to a weight function $w$ is a subset $S\subset V$ of vertices such
that there is a partition $V=S\cup V_1 \cup V_2$
such that $w(V_1), w(V_2)\le 2/3$ and there are no edges between $V_1$ and $V_2$. If the weight function is not specified, it is assumed that $w(v) = \frac{1}{|V|}$, for every vertex $v$.
We refer to Arora and Barak~\cite{arora2009computational} for a basic understanding of the model and terminologies for space-efficient algorithms.

\section{Intersection Graphs of Jordan Regions}\label{jordon}

In this section, we study the \textsc{Reachability} problem on
the intersection graphs of Jordan regions. Let $\mathcal{C}$ be a set such that each element $C$ of $\mathcal C$ is a simply connected compact region in a plane bounded by a closed Jordan curve. Let $G(\mathcal{C})$ be an intersection graph on $\mathcal{C}$, where two distinct elements $C_1,C_2\in \mathcal{C}$ are adjacent if and only if their intersection is not empty. Additionally,
each edge of $G(\mathcal{C})$ is a directed edge. A Jordan region $A$ \emph{contains} another Jordan region $B$ if $A \subseteq int(B)$, where $int(B)$ is the interior of $B$.
We assume that no point is a boundary point of three elements of $C$, and each element of $C$ intersects at least one other element of $C$. Further, we may assume that the number of such intersection points is finite.
In~\cite{fox2008separator}, Fox and Pach showed the existence of a $O(m^{1/2})$-size separator on a collection of Jordan regions, where $m$ is the total number of crossings points in the boundary of Jordan regions.

For computation, we work with those classes of Jordan regions that can be represented compactly, i.e., the set of $n$ Jordan regions can be represented by poly$(n)$ bits. Furthermore, we assume that basic operations, such as determining the intersection point of two Jordan regions and outputting a constant number of points of any Jordan region, can be performed in log-space.
We prove the following theorem about the separator. Note that any \emph{vertex separator} of an intersection graph does not rely on the direction of the edges. Hence for ease of explanation, we drop the directions from the input graph $G(\mathcal{C})$.

\begin{theorem}\label{jordan-sep-th}
Let $\mathcal{C}$ be a collection of Jordan regions, and let $G(\mathcal{C})$ be an intersection graph on $\mathcal{C}$. Let $w$ be a weight function on $\mathcal{C}$. There exists a polynomial time algorithm that takes as an input the set $\mathcal{C}$ and outputs a separator of $G(\mathcal{C})$ of size $O(m^{1/2})$ using $O(m^{1/2}\log n)$ space.
\end{theorem}

\begin{figure}
\centering
\includegraphics{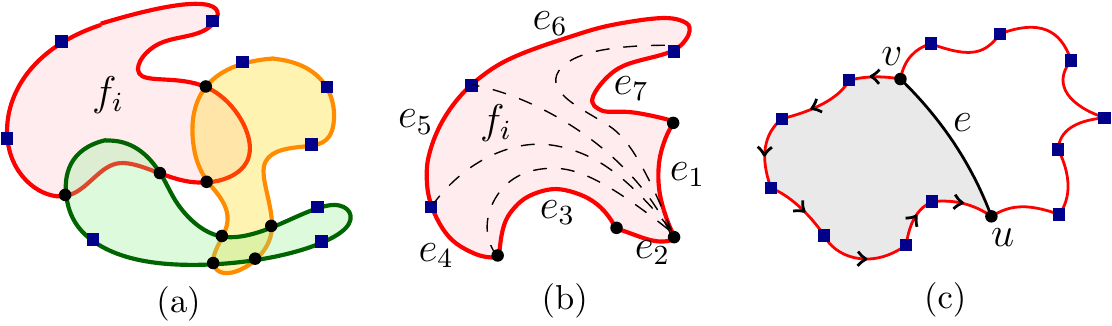}
\caption{(a) An illustration of the planar embedding of the curves in $\mathcal{C}$. The black disk points are the intersection points from set $A(\mathcal{C})$, and the blue square points the three extra points from set $B(\mathcal{C})$. (b)
A illustration of a face in the embedding and its corresponding planar triangulation. (c) Traversal of the left-face of the edge $e=(u,v)$.}
\label{jordan-fig}
\end{figure}

We assume that the sum of the weights of the given Jordan regions is one.
Let $H(\mathcal{C})$

be the set of \emph{heavy} regions in $\mathcal{C}$ whose weight is more than $1/m^{1/2}$, and let $L(\mathcal{C})$ be set of regions in $\mathcal{C}$
that are involved\footnote{We say that a region $C_1$ is involved in a \emph{containment} with a region $C_2$ if either $C_1$ contains $C_2$ or $C_2$ contains $C_1$.} in at least $\frac{1}{3}m^{1/2}$ containments with other elements of $\mathcal{C}$.
Let $I(\mathcal{C}) = \mathcal{C} \setminus (H(\mathcal{C}) \cup L(\mathcal{C}))$.
We define a planar graph $G_{P}(\mathcal{C})$.
The vertex set of this graph is the union of two subsets, i.e., $A(\mathcal{C}) \cup B(\mathcal{C})$, where $A(\mathcal{C})$ is the set of all intersection points that lie on the boundary of at least one element of $I(\mathcal{C})$ and $B(\mathcal{C})$ is a collection of points not in $A(\mathcal{C})$ such that the boundary of each $C \in I(\mathcal{C})$ contains precisely three points in $B(\mathcal{C})$. There exists an edge between two vertices of $G_P(\mathcal{C})$ if and only if there are consecutive points along the boundary of an element of $I(\mathcal{C})$; see Figure~\ref{jordan-fig}(a).

To see that the defined graph is planar, we see that we have added a vertex at every intersection point of the boundary of the given Jordan curves. We can now draw the edges along the boundary. We proceed with the following lemmas.

\begin{lemma}\label{jordan-lemma-1}
There exists a log-space algorithm that takes as an input a set of Jordan regions $\mathcal{C}$ and outputs a graph $G_P(\mathcal{C})$.
\end{lemma}
\begin{proof}
First of all, note that there exists a log-space subroutine that on an input Jordan region $C$ of $\mathcal{C}$ determines if $C$ is in $L(\mathcal{C})$. To do this, the subroutine iterates over all Jordan regions in $\mathcal{C}$ and calculates the total number of containment with which $C$ is involved.
The weight of each Jordan region comes with the input. Hence it is enough for us to determine if a Jordan region $C$ of $\mathcal{C}$ belongs to the set $I(\mathcal{C})$. In order to construct the planar graph, we essentially need to know $I(\mathcal{C})$, thereby concluding the proof of the lemma.
\end{proof}

Next, we triangulate $G_{P}(\mathcal{C})$, and denote the triangulated planar graph by $G_{T}(\mathcal{C})$. This triangulation can be obtained in log-space using the following lemma

\begin{lemma}\label{jordan-triangle}
There exists an algorithm that takes a planar graph in the input and returns a triangulated planar graph in log-space.
\end{lemma}

\begin{proof}
Consider a planar graph $G_P$. The embedding of a planar graph can be computed in log-space~\cite{elberfeld2014embedding}. Note that an edge $e_1=(u_1,v_1)$ in
$G_P$ is part of two faces in the planar embedding of the graph. We call them the left and the right face, respectively. We explain how to traverse the left face of $e_1$ in log-space. In order to do this, we start by traversing $e_1$ to one of its endpoint (say $v_1$) and take the edge clockwise next to $e_1$, that is incident on $v_1$. Let $e_2$ be such an edge; see Figure~\ref{jordan-fig}.
We continue the traversal along the edge $e_3$, which is the edge clockwise next to $e_2$ from the endpoint $v_2$ of $e_2$, and so on. In general, if we reach at the vertex $v_i$ using the edge $e_i$, we continue along the edge clockwise next to $e_i$ from $v_i$. Clearly, by following this procedure, we can traverse the boundary of the face. In order to prove this lemma, it is sufficient to show that given two vertices $u_l$ and $v_l$ of the input planar graph $G_P$, whether an edge $e_l$ can be added between them in log-space as part of the triangulation.
We assume that the vertices of the input graph $G_P$ are indexed by an integer from $1$ to $k$, for some $k\in [n]$. For each edge $e$ that is incident with $u$, we first traverse the left face of $e$ and see if (i) $v$ is present in that face, (ii) either $u$ or $v$ is the lowest-indexed vertex of that face. There is a triangulated edge between $u$ and $v$ if and only if both these conditions hold for any such edge $e$; see Figure~\ref{jordan-fig}(b). This concludes the proof of the lemma.
\end{proof}

Let $d(C)$ be the number of points on the boundary of $C$ that belong to the vertex set of $G_{T}(\mathcal{C})$. For a vertex $v$ in $G_{T}(\mathcal{C})$, we define a new weight function $\weight{v}$ as follows:

\noindent \textbf{Case 1:} $v$ is part of two boundaries $C_1$ and $C_2$ of $\mathcal{C}$, $\weight{v} = \frac{w(C_1)}{d(C_1)} + \frac{w(C_2)}{d(C_2)}$.

\noindent \textbf{Case 2:} $v$ is in the boundary of only one element $C_1$ of $\mathcal{C}$, $\weight{v} = \frac{w(C_1)}{d(C_1)}$.

Fox and Pach~\cite{fox2008separator} used the idea to find a cycle-separator in this triangulated graph $G_{T}(\mathcal{C})$ and used it to construct a separator of the original geometric intersection graph. Instead, we will use the result of Imai et al.~\cite{DBLP:conf/coco/ImaiNPVW13} to obtain such a separator in a space-efficient manner.

\begin{lemma}[\cite{DBLP:conf/coco/ImaiNPVW13}]
Let $G$ be a triangulated planar graph. There exists a polynomial time algorithm that uses $O(\sqrt{n}\log n)$ space to output a separator of $G$ of size $O(\sqrt{n})$.
\end{lemma}

Now, consider the set of regions of $\mathcal{C}$ whose boundary contains at least one of the points of the separator returned by the algorithm of Imai et al. on $G_T(\mathcal{C})$.
We denote these regions by $\sep{\mathcal{C}}$.
In the following, we show that this set is indeed the required separator. Moreover, this set can be calculated within the required space-time bounds.

\begin{lemma}\label{jordan-sep}
The set $\sep{\mathcal{C}}$ is a separator of the intersection graph of $\mathcal{C}$
\end{lemma}
\begin{proof}
First of all, note $G_{T}(\mathcal{C})$ is a triangulated planar graph. Let $S$ be the separator of this graph returned by the algorithm of Imai et al.~\cite{DBLP:conf/coco/ImaiNPVW13}. For a triangulated graph, we observe that the separator $S$ is a cycle. Let $V_0$ be the set of all elements of $\mathcal{C}$ that are either not in $I(\mathcal{C})$, or whose boundary contains a vertex of $S$. Let $K_1$ and $K_2$ be the set of vertices that are \emph{inside} and
\emph{outside} the cycle $S$, respectively.
Let $V_i$ (for $i\in\{1,2\}$) be the set of elements of $I(\mathcal{C})$, such that all vertex of $V(G_{T}(\mathcal{C}))$ which are on its boundary belong to the component $K_i$. It is easy to see that $V_0$, $V_1$, and $V_2$ are pairwise disjoint sets, and their union is $\mathcal{C}$. If $V_0$ has a weight of at least $1/3$, we see that it acts as a trivial separator. Hence, for the rest of this proof, we assume that the weight of $V_0$ is less than $1/3$. Thus, we can also see that the weight of $V_i$ is at most $2/3$. It only remains to show that there is no edge between a vertex of $V_1$ and a vertex of $V_2$.

Let us assume, w.l.o.g., that the weight of $V_2$ is greater than the weight of $V_1$ and, therefore, greater than $1/3$. In order to show that there is no edge between a vertex of $V_1$ and a vertex of $V_2$, let us assume to the contrary that there is one such edge. Since $V_1$ and $V_2$ are on two different sides of a closed Jordan curve, there exists an element $v$ of $V_1$ that contains in its interior all the elements of $V_2$. The number of elements in $V_2$ is at least $\frac{1}{3}m^{1/2}$ contradicting the fact that $v$ belongs to $I(\mathcal{C})$.
\end{proof}

Next, in the final lemma, we argue about the space-complexity of the procedure mentioned above.

\begin{lemma}\label{jordan-space}
There exists a polynomial time algorithm that takes as an input $\mathcal{C}$ and outputs the set $\sep{\mathcal{C}}$ in $O(m^{1/2}\log n)$ space.
\end{lemma}
\begin{proof}
We first see that on input $\mathcal{C}$, the graph $G_{T}(\mathcal{C})$ contains $O(m)$ vertices. We can output this graph in log-space (see Lemma~\ref{jordan-triangle}). Then, using the planar separator algorithm of Imai et al.~\cite{DBLP:conf/coco/ImaiNPVW13}, we can get a set of $O(m^{1/2})$ vertices which acts as the separator of this graph. This process can be done in $O(m^{1/2}\log n)$ space and polynomial time. Once we obtain this set, it is possible to construct $\sep{\mathcal{C}}$ by using this set.
\end{proof}

This completes the proof of Theorem~\ref{jordan-sep-th}. Then, by combining Theorem~\ref{jordan-sep-th} and Theorem~\ref{treewidth}, we conclude the following.

\begin{corollary}
There exists an algorithm that solves the \textsc{Reachability} on directed intersection graphs of Jordan regions in polynomial time and $O(m^{1/2}\log n)$ space.
\end{corollary}

\section{Unit Contact Disk Graphs (Penny Graphs)}\label{penny}
We now study the \textsc{Reachability} problem on unit contact disk graphs (penny graphs). Penny graphs are also known as unit coin graphs~\cite{cerioli2011note}. The vertices of a penny graph are unit circles in the plane, such that no two of those circles cross each other, and there is an edge between two vertices if and only if the corresponding circles touch each other.
We prove the following theorem.

\begin{theorem}\label{th:penny}
For every $\varepsilon> 0$, there exists a polynomial time algorithm that can solve \textsc{Reachability} in an $n$ vertex directed penny graph, using $O(n^{1/4 + \varepsilon})$ space.
\end{theorem}

We know that given a set of $n$ unit disks with $m$ intersections, there exists an axis-parallel line intersecting $O(\sqrt{m+n})$ disks such that each half-plane separated by that line contains at most $4n/5$ disks~\cite{DBLP:journals/comgeo/CarmiCKKORRSS20}.
They call such a separator a \emph{balanced separator}.
We first describe that how such a balanced separator can be obtained in a space-efficient manner.

Let $G=(D,E)$ be a directed penny graph, where $D=\{d_1,\ldots,d_n\}$ is a set of unit disks, and there is an edge $e\in E$ between two disks $d_i$ and $d_j$ (for some $i,j\in [n]$) if $d_i$ and $d_j$ touch each other. Moreover, each edge in $e\in E$ is a directed edge.
Let $R$ be a rectangular bounding box that contains the disks in $D$. Our algorithm to solve the \textsc{Reachability} problem on penny graphs consists of four steps: (1) we find the balanced separators and use them to obtain a rectangular subdivision of the plane, (2) create an auxiliary graph,
(3) obtain a pseudo-separator, (4) solve the \textsc{Reachability}.

\subsection{Obtaining Rectangular Subdivion}
In this step, the directions of the edges of $G$ are not relevant, so we consider the input as a set $D$ of unit non-overlapping unit disks.
The main idea is to divide the rectangle $R$ into smaller rectangles such that each smaller rectangle contains at most $n^{1-\varepsilon}$ unit disks, and the number of unit disks intersecting the boundary of any rectangle is at most $n^{1/2 + \varepsilon/2}$.

In what follows, we describe a sweeping procedure.
We discuss primarily in terms of vertical sweeping, and horizontal sweeping is done similarly.

We consider the disks in sorted order based on the x-coordinates of their corresponding centers, which can be done in log-space.
Consider a vertical sweep line $\ell$.
We start with the leftmost disk from $D$ and sweep until we find a balanced separator that is intersecting $O(\sqrt{m+n})$
many disks such that each half-plane separated by that line contains at most $4n/5$ disks. If we do not find any such vertical separator, we apply the same procedure with a horizontal sweep line (say $\ell'$) sweeping from top to bottom.
From Theorem~2 in \cite{DBLP:journals/comgeo/CarmiCKKORRSS20},
we know that there exists an axis-parallel balanced separator. Therefore, we shall obtain a balanced separator by doing this procedure.

When the sweep line $\ell$ encounters or leaves a disk, we call it an event. The number of \emph{events} is precisely $2n$. Note that it is possible to test by using $\log n$ space whether a disk $d_i$ intersects the line $\ell$. We maintain a counter $c_{\ell}$ corresponding to the sweep line $\ell$. At each event $k$ (for some $k\in [n]$), we determine the number of disks that intersect $\ell$, by checking each disk that whether it is intersecting with $\ell$ or not in log-space.
Then when the next event happens, we increase the counter's value by $1$ if it intersects a new disk, or decrease it by $1$, otherwise. Using this procedure, we can find a separator line $\ell$, which is a balanced separator. It is also clear that we can determine such a separator by using $\log n$ space. Once we find such a separator, we only store the $x$-coordinate (resp. $y$-coordinate) of the vertical line $\ell$ (resp. horizontal line $\ell'$).
Note that it is possible to find the actual set of disks that form such a separator in log-space when needed.

We subdivide the rectangles repeatedly until each of the rectangles has smaller than $n^{1-\varepsilon}$ disks. Initially, we have the rectangle $R$ containing all the disks. Let $\mathcal{R}_0=\{R\}$ be the initial set of rectangles.
After step $i$, we have the rectangles $\mathcal{R}_i$ be the set of rectangles, and We pick the rectangle with more than $n^{1-\varepsilon}$ disks and subdivide it further using the above process to get $\mathcal{R}_{i+1}$. See Figure~\ref{grid-penny} for an illustration.

\begin{figure}[ht]
\centering
\includegraphics{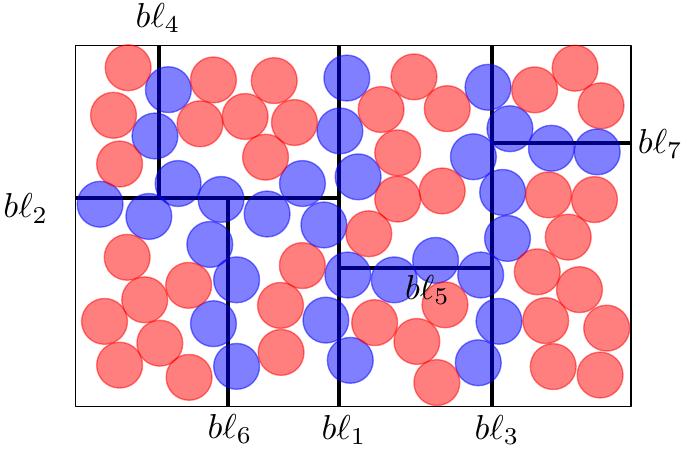}
\caption{An illustration of the rectangular subdivision by using the balanced separators. The blue disks are the ones intersected by the balanced separator lines, and the red disks are contained inside the rectangles.}
\label{grid-penny}
\end{figure}

We calculate the number of separators required to reach this termination point. From~\cite{DBLP:journals/comgeo/CarmiCKKORRSS20}, we know that one $\sqrt{m+n}$ size separator guarantees that on each side there are at most $\frac{4n}{5}$ disks. Since the class of penny graphs is a subclass of planar graphs, the total number of edges is at most $3n-6$. Now at each step, we have obtained a \emph{balanced separator} whenever it has satisfied the criteria. We need $O(n^{\varepsilon})$ many separators to have at most $O(n^{1-\varepsilon})$ disks in each cell.
To store these line separators, we need to use $O(n^{\varepsilon}\log n)$ space.

The initial graph $G$ is divided into $n^{\varepsilon}\times n^{\varepsilon}$ rectangles obtained from the above procedure. Let $\mathcal{Z}$ be the set of all rectangles. The idea is to reduce the size of the graph $G$ by dropping the disks that are entirely contained inside some rectangle and are not touched or intersected by its boundary line.
However, while reducing the size of the number of disks, we need to ensure that the \textsc{Reachability} information is fully preserved between any pair of disks in $G$. For that, we proceed to the next step and build an auxiliary graph.

\subsection{Building Auxiliary Graph}
For a rectangle $R$, let $G_R$ be the graph defined as follows. The vertex set of $G_R$ is the set of all disks which intersect at least one of the boundaries of $R$. We add an edge from a vertex $u$ to a vertex $v$ in $G_R$ if there is a directed path from $u$ to $v$, which contains only the disks present inside the rectangle $R$. Let $v_1$ be an arbitrary disk, and let $\{v_1,\ldots,v_k\}$ be the sets of disks intersecting the boundary of $R$ in the anti-clockwise order. We place the disk centers on the boundary while preserving (1) the order of them on the boundary, (2) each vertex $v_i$ is on the side of the boundary that intersects it. However, if it intersects by more than one side (one vertical and one horizontal), we create an additional dummy vertex (say $v'_i$) and assign
$v_i$ and $v'_i$ to the horizontal and vertical side, respectively. Moreover, we add a bidirectional edge between $v_i$ and $v'_i$. see Figure~\ref{aux-penny}(b)) for an illustration.

The edges of $G_R$ are drawn in the following manner. If there is an edge between two vertices that are on different sides then we give a directed straight line edge. Otherwise, we join them by a directed circular arc. Moreover, we ensure that there is a crossing between two edges $(v_i,v_k)$ and $(v_l,v_j)$ in the drawing if and only if there is an ordering on the boundary which is one of the followings - 1. $\{v_i\prec v_l\prec v_k \prec v_j\}$, 2. \{$v_j\prec v_l\prec v_k \prec v_i$\}, 3. \{$v_l\prec v_i\prec v_j \prec v_k$\}, 4. $\{v_l\prec v_i\prec v_j \prec v_k\}$. There exists such a drawing as their arc edges can be drawn arbitrarily close to the boundary lines.

\begin{figure}[ht]
\centering
\includegraphics{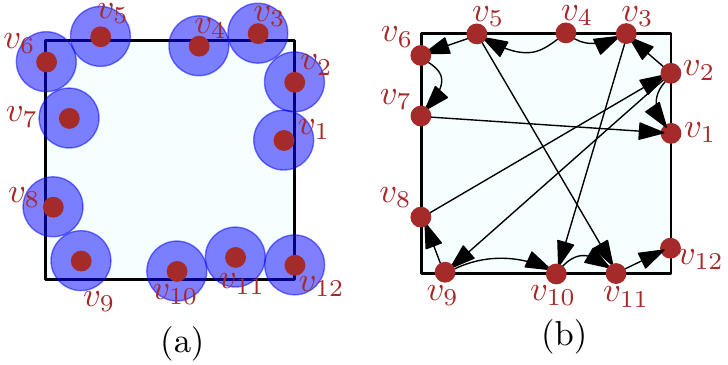}
\caption{An illustration of the representation of the graph $G_R$, (a) the blue disks are the vertices of $G_R$ intersecting the boundary of the rectangle (b) drawing of $G_R$.}
\label{aux-penny}
\end{figure}

Now, by combining the graphs defined for each rectangle, we define the auxiliary graph $Aux_{\varepsilon}(G)$, for $0<\varepsilon<1$.
The vertex set of $Aux_{\varepsilon}(G)$ is $\bigcup_{R\in \mathcal{Z}} V(G_R)$ and the edge set is $\bigcup_{R\in \mathcal{Z}} E(G_R)$. Notice that
$Aux_{\varepsilon}(G)$ might have parallel edges since there exist paths between vertices in two adjacent rectangles, and in that case, we keep both of these edges in their respective rectangles. The total number of vertices in each cell is $O(n^{1-\varepsilon})$. Hence the total number of vertices in $Aux_{\varepsilon}(G)$ is $O(n^{1/2+\varepsilon/2})$.
We point out that we do not store $Aux_{\varepsilon}(G)$ explicitly because that requires too much space. Instead, we deal with each cell recursively when the subroutine queries for an edge in that cell of $Aux_{\varepsilon}(G)$.
Now, we prove the following property about the auxiliary graph.

\begin{lemma}\label{lemma-penny-1}
Let $G$ be a penny graph and $e=(v_i,v_j)$ and $e'=(v_k,v_l)$ be two edges in $Aux_{\varepsilon}(G)$. If $e$ and $e'$ cross each other, then $Aux_{\varepsilon}(G)$ also contains the edges $(v_i,v_l)$ and $(v_k,v_j)$.
\end{lemma}
\begin{proof}
Consider two edges $e=(v_i,v_j)$ and $e'=(v_k,v_l)$ be two edges in $Aux_{\varepsilon}(G)$. From the definition, these edges corresponding to the directed paths in the input graph $G$. From the construction, we know that the ordering of the end vertices are one of the following -
$\{v_i\prec v_l\prec v_k \prec v_j\}$, \{$v_j\prec v_l\prec v_k \prec v_i$\}, \{$v_l\prec v_i\prec v_j \prec v_k$\}, $\{v_l\prec v_i\prec v_j \prec v_k\}$. Now, we know that the corresponding directed paths are fully embedded inside the grid cell. If two directed path intersects in the Auxiliary graph (based on our definition) of an input penny graph, that is embedded inside a rectangle, then they must have one common vertex;
see Figure~\ref{aux-lemma-penny}. Otherwise, it will not admit a planar embedding. Hence, there must be directed path from $v_i$ to $v_l$, and $v_k$ to $v_j$. Thereby proving the lemma.\end{proof}
\begin{figure}[ht]
\centering
\includegraphics{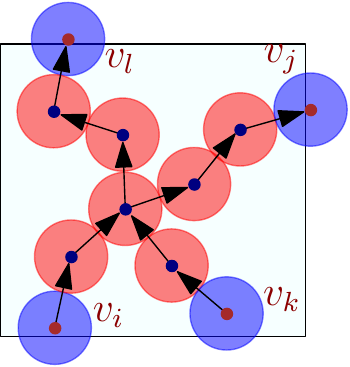}
\caption{An illustration of two directed paths inside a grid cell.}
\label{aux-lemma-penny}
\end{figure}

\subsection{Constructing Pseudo-Separator}
An essential property of a vertex separator is that, for any two vertices $u$ and $v$, a path between them must contain a vertex from the separator if $u$ and $v$ lie in two different components with respect to the separator.

We use a separator construction for the auxiliary graph $Aux_{\varepsilon}(G)$. However, note that $Aux_{\varepsilon}(G)$ is not a planar graph anymore. Therefore, we need a special kind of separator, which we call a \emph{pseudo-separator}.

The notion of the pseudo-separator was introduced by Jain and Tewari~\cite{DBLP:conf/fsttcs/JainT19} in the context of grid graphs.
However, since the class of Penny graphs is a superclass of grid graphs, it is not possible to use their idea directly.

Let $G$ be a penny graph and $H=(V_1,E_1)$ be a vertex induced subgraph of $Aux_{\varepsilon}(G)$ with $h$ vertices. Let $f:\mathbb{N} \rightarrow \mathbb{N}$ be a function. A subgraph $C=(V_2,E_2)$ of $H$ is said to be an $f(h)$-PseudoSeparator of $H$, if the size of every connected component in $H \cdot C$ is at
most $f(h)$, where the graph $H\cdot C=(V_3,E_3)$ is defined as $V_3=V_1\setminus V_2$ and $E_3=E_1\setminus \{e\in E_1 \mid \exists e'\in E_2, e$ crosses $e'\}$. The size of $C$ is the total number of vertices and edges of $C$ summed together.

The general idea of our approach is the following.
Consider a vertex induced subgraph $H$ of $Aux_{\varepsilon}(G)$.
We choose a maximal subset of edges such that $H$ is a planar graph, thus admits a planar embedding.
Next, we triangulate this chosen sub-graph.
We show that until this point, each operation can be performed in log-space.
Then, we use the algorithm of Imai et al.~\cite{DBLP:conf/coco/ImaiNPVW13} to obtain a separator of the triangulated graph. In what follows, we describe these procedures in detail.

We start with a maximal planar graph $H(Aux_{\varepsilon}(G))$ of $Aux_{\varepsilon}(G)$. The vertex set of $H(Aux_{\varepsilon}(G))$ is same as the vertex set of $Aux_{\varepsilon}(G)$. For each rectangle, we index the vertices in ascending order while traversing them in anti-clockwise direction. For each rectangle $R$, an edge $e_l=(u_l,v_l)$ of $\rect{G}{\rho}$ is added to $H(Aux_{\varepsilon}(G))$ if there is no edge $e_k=(u_k,v_k)$ such that the ordering of the vertices is one of the following -
$\{u_k \prec u_l \prec v_k \prec v_l\}$,
$\{v_k \prec u_l \prec u_k \prec v_l\}$,
$\{u_k \prec v_l \prec v_k \prec u_l$\},
$\{u_k \prec v_l \prec v_k \prec u_l\}$.
We prove that,
$H(Aux_{\varepsilon}(G))$ is indeed a maximal planar graph of $Aux_{\varepsilon}(G)$.

\begin{lemma}[$\star$]\label{claim-penny}
$H(Aux_{\varepsilon}(G))$ is a maximal planar graph of $Aux_{\varepsilon}(G)$.
\end{lemma}
\begin{proof}
Note that $H(Aux_{\varepsilon}(G))$ is a planar graph that comes from the construction. Since for each rectangle $R$, we have selected edges from $\rect{G}{\rho}$ such that no two edges intersect each other.
We prove the maximality by contradiction.
Assume that there exists an edge $e_l=(u_l,v_l)$ in
$Aux_{\varepsilon}(G)$ that is not chosen in $H(Aux_{\varepsilon}(G))$, and this edge does not intersect any edge in $H(Aux_{\varepsilon}(G))$. Let $R^*$ be the rectangle that contains $e_k$.
We define the width of an edge $e=(u,v)$ as the smallest number of vertices that we encounter between $u$ and $v$ along the rectangle boundary. We pick the edge $e_b=(u_b,v_b)$ with the largest width,
whose exactly one endpoint lies in between $u_l,v_l$ in the total ordering of the boundary vertices (see Figure~\ref{maximal-planar-penny}). From the definition of $H(Aux_{\varepsilon}(G))$ we know that there exists be an edge.

Now, if this edge is a directed edge from $u_b$ to $v_b$,
and such an edge is not chosen in $H(Aux_{\varepsilon}(G))$ because there is another edge $e_g=(u_g,v_g)$ such that $u_g$ lies between $u_b$ and $u_l$, then by applying Lemma~\ref{lemma-penny-1} we can argue that there is an edge between $u_g$ to $v_b$ with larger width.
This contradicts the fact that the edge $(u_b, v_b)$ has the largest width among those which intersects $(u_l, v_l)$.
See Figure~\ref{maximal-planar-penny} for an illustration.
Moreover, based on the direction of the edges,
there are other uniform cases that we do not discuss explicitly.\end{proof}

\begin{figure}[ht]
\centering
\includegraphics{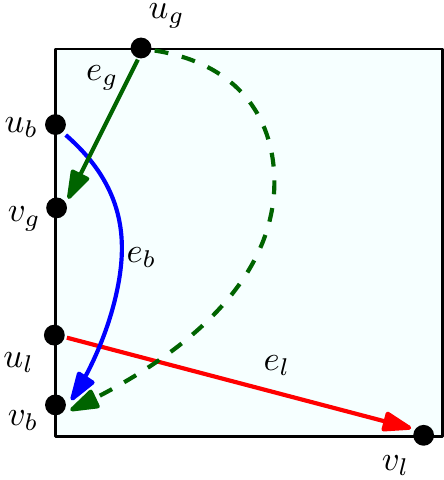}
\caption{An illustration of the proof of Lemma $\ref{claim-penny}$}
\label{maximal-planar-penny}
\end{figure}

Next, we triangulate $H(Aux_{\varepsilon}(G))$ by adding the boundary edges.
Then for each rectangle, we consider a face of $H(Aux_{\varepsilon}(G))$ and add the edges to complete the triangulation by a similar procedure described in Lemma~\ref{jordan-triangle} from Section~\ref{jordon}. Moreover, the directions of the edges are arbitrary.
Let $\hat{H}(Aux_{\varepsilon}(G))$ be the triangulated graph.
Note that this procedure can be done in log-space. Next, we use a Lemma from~\cite{DBLP:conf/coco/ImaiNPVW13} that is stated below.

\begin{lemma}[\cite{DBLP:conf/coco/ImaiNPVW13}]\label{imai-penny}
For each $\beta>0$, there exists a polynomial time algorithm and $\tilde{O}(h^{1/2+\beta/2})$ space algorithm that takes a $h$-vertex planar graph $P$ as input and and outputs a set of vertices $S$, such that $|S|$ is $O(h^{1/2+\beta/2})$ and removal of $S$ disconnects the graph into components of size $O(h^{1-\beta})$.
\end{lemma}

Next, we construct the pseudo-separator $S_H(Aux_{\varepsilon}(G))$ by using the following steps. First, by Lemma~\ref{imai-penny} we find a set $S$ in $\hat{H}(Aux_{\varepsilon}(G))$ that divides it into components of size $O(h^{1-\beta})$, where $h$ is the number vertices in $\hat{H}(Aux_{\varepsilon}(G))$.
We add the vertices and edges of $S$ to the vertex and edge set of $S_H(Aux_{\varepsilon}(G))$, respectively.
However, there is a small caveat to use Lemma~\ref{imai-penny} on $\hat{H}(Aux_{\varepsilon}(G))$. In order to triangulate the graph, we have added edges that were originally not part of the auxiliary graph. Therefore, for each edge $e_k=(u_k,v_k)$ of the triangulation that is present in some rectangle $\mathcal{R}$, we consider a set of at most four edges of $Aux_{\varepsilon}(G)$ that form a so-called \emph{shield} around the edge $e_k$. Two of these edges start from $u_k$ ending at two vertices $v_p,v'_p$, where $v_p$ and $v'_p$ are the closest points to the left and the right of $u_k$, respectively, in the total ordering of the boundary vertices. The other two edges start from $v_k$ and end at two vertices $u_q,u'_q$, where $u_q$ and $u'_q$ are the closest points to the right and the left of $v_k$, respectively, in the total ordering of boundary vertices. See Figure~\ref{pseudo-separaor-penny}(a) for an illustration.
Later, we argue that if these edges do not exist, and the cycle separator intersects the corresponding triangulation edge, there must be other edges chosen in the maximal planar graph intersecting that triangulation edge, and hence this gives a contradiction.

In order to prove that $S_H(Aux_{\varepsilon}(G))$ is indeed a pseudo-separator, we need a property of triangulated graphs from~\cite{DBLP:conf/isaac/JainT19}.

\begin{lemma}[\cite{DBLP:conf/isaac/JainT19}]\label{lemma-triangle-penny}
Let $G$ be a triangulated planar graph and $S$ be a subset of its vertices. For every pair of vertex $u,v$ which belong to different components of $G\setminus S$, there exists a cycle in $G[S]$, such that $u$ and $v$ belong to different sides of this cycle.
\end{lemma}

Next, we prove the Lemma about pseudo-separator.

\begin{lemma}[Pseudo-Separator Lemma]\label{lem:psep}
Let $G$ be a penny graph and $H(Aux_{\varepsilon}(G))$ be a vertex induced subgraph of $Aux_{\varepsilon}(G)$. The
graph $S_H(Aux_{\varepsilon}(G))$ is a $h^{1-\beta}$-pseudo-separator of $H(Aux_{\varepsilon}(G))$.
\end{lemma}
\begin{proof}
Let $S$ be a set of vertices chosen from $\hat{H}(Aux_{\varepsilon}(G))$ by using Lemma~\ref{imai-penny}. The claim is that if two vertices $u$ and $v$ belong to two different components in $\hat{H}(Aux_{\varepsilon}(G))\setminus S$, then a path between $u$ and $v$ in $G$ either takes a vertex of $S_H(Aux_{\varepsilon}(G))$ or crosses an edge of $S_H(Aux_{\varepsilon}(G))$. Note that due to Lemma~\ref{lemma-triangle-penny}, we know that $u$ and $v$ are on two different sides of a cycle of $S$. Thus, any edge drawn in the plane connecting $u$ and $v$ (not necessarily a straight line edge) crosses the cycle of $S$. If the path crosses non-triangulation edge or an edge which is shielded properly (see Figure~\ref{pseudo-separaor-penny}(a)), we are done as the path will intersect an edge of $S_H(Aux_{\varepsilon}(G))$.
The other situation is when such a path intersects a triangulation edge. We prove by contradiction that this can not happen. Assume that it does happen (see Figure~\ref{pseudo-separaor-penny}(b)). Let $e_r$ be the triangulation edge intersected by the edge $e_b$ of the path from $u$ to $v$. Now, since $e_b$ is not chosen in the maximal planar graph, this means there exist other edges chosen in the $H(Aux_{\varepsilon}(G))$. Let $e_o$ and $e_o'$ be two such edges with maximum width on either side of the triangulated edge (see the orange edges in Figure~\ref{pseudo-separaor-penny}(b)). By Lemma \ref{lemma-penny-1}, there exists an edge $e_g$ from the tail of $e_o$ to the head of $e_o'$
(the green path in Figure~\ref{pseudo-separaor-penny}(b)). Since this edge $e_g$ also crosses the triangulation edge $e_r$, it is not present in the maximal planar graph. Any edge that is present in maximal planar graph and crosses $e_g$ contradicts the fact that $e_o$ and $e_o'$ are of maximum width. There are other cases when either $e_o$ or $e_o'$ is not present, but they can be handled similarly.
Thereby concluding the proof of the Lemma.
\end{proof}

\begin{figure}[ht]
\centering
\includegraphics{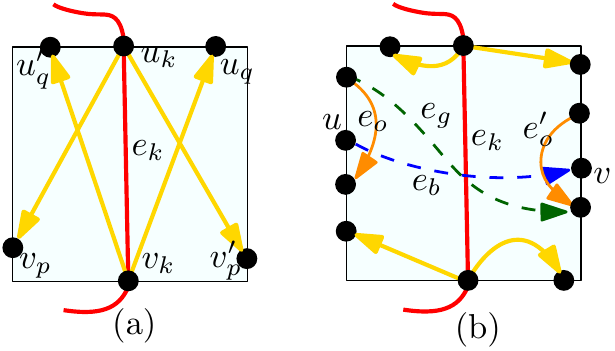}
\caption{An illustration of the pseudo-separator. The red edge is an edge of the triangulation and part of the red cycle separator. The yellow edges are chosen to form the shield (a) when the proper shield exists (b) otherwise.}
\label{pseudo-separaor-penny}
\end{figure}

\subsection{The Algorithm} Let $H$ be a vertex induced subgraph of an auxiliary graph. We first explain how to solve \textsc{Reachability} in $H$. Initially, $H$ is the whole auxiliary graph, and we wish to find the \textsc{Reachability} between given two vertices $s$ and $t$ of $H$.

By using Lemma \ref{lem:psep}, we find a pseudo-separator $S$ of $H$. W.l.o.g., assume that $s$ and $t$ are both in $S$. The pseudo-separator divides the graph into components $C_1, C_2, \ldots, C_k$, for $k\in [h]$. We maintain a set
of vertices $M$ that we call \emph{marked} vertex set. We use an array of size at most $\lvert S \rvert$ to mark a set of vertices $M_1$ of the pseudo-separator. Additionally, for each edge of the pseudo-separator we have at most one associated vertex, and they form the set $M_2$. Then, $M=M_1\cup M_2$ is called a \emph{marked} vertex set.

Throughout the algorithm, we maintain that if a vertex is marked then there is a path from $s$ to that vertex in the auxiliary graph. Initially, only the vertex $s$ is marked. We then perform $h$ iterations. In each iteration, we update the set of marked vertices as follows:

\noindent\textbf{Step~1.} For every vertex of $S$, we mark it if there is a path from an already marked vertex to it such that the internal vertices of that path all belong to only one component $C_i$, for $i\in [h]$. We check this by recursively running our algorithm on the subgraph of $H$ induced by the vertex set $C_i$.

\noindent\textbf{Step~2.} For each edge $e$ of $S$, the algorithm sets the associated marked vertex to $u$ if the following three conditions are satisfied: (a) There exists an edge $e'= (u, v)$ which crosses $e$, (b) there exists a path from a marked vertex to $u$ such that the internal vertices of that path all belong to only one component $C_i$ (we check this by recursively running our algorithm on the vertex induced subgraph of $C_i$), (c) $e'$ is the closest such edge to $u$.

Finally, we output 'YES' if and only if $t$ is marked at the end of these iterations.

\subsubsection{Correctness}
Let $P$ be a path from $s$ to $t$ in $H$. Suppose $P$
passes through the components $C_{\sigma_1}, C_{\sigma_2}, \ldots, C_{\sigma_l}$ (each $\sigma_i$ is a value in $[k]$)
in this order. By the definition of a pseudo-separator, the path can go from one component, $C_{\sigma_i}$ to the next component $C_{\sigma_{i+1}}$ in the following ways only:

\noindent\textbf{Case 1:} The path exits $C_{\sigma_i}$ and enters $C_{\sigma_{i+1}}$ through a vertex $w$ of the pseudo-separator.

\noindent\textbf{Case 2:} The path exits $C_{\sigma_i}$ and enters $C_{\sigma_{i+1}}$ through an edge $e' = (u, v)$ whose tail is in $C_{\sigma_i}$ and head is in $C_{\sigma_{i+1}}$. This edge will cross an edge $e = (x, y)$ of the pseudo-separator.

We see that after the $i$-th iteration, our algorithm will traverse the fragment of the path in the component $C_{\sigma_i}$ and either mark in \textbf{Case 1} its endpoint or in \textbf{Case 2} a vertex $u'$ such that the edge $(u', v)$ exist. Thus, $t$ will be marked after $l$ iterations if and only if there is a path from $s$ to $t$ in $H$. Since $l$ can be at most $h$, these many iterations would suffice.

From the above discussion, we know how to solve \textsc{Reachability} on the auxiliary graph with the desired space-bound. To solve the problem on the input graph, we find the rectangular subdivision of the input graph. Then, we form the auxiliary graph by solving each rectangle recursively. Next, we solve the \textsc{Reachability} problem on the auxiliary graph by following the above procedure. We have shown earlier that the reachability information of the input graph is fully preserved in the auxiliary graph. Hence, it is possible to report the solution of the input directed penny graph from the auxiliary graph's solution.

\subsubsection{Space-Complexity} The space complexity of our algorithm is dominated by the space required to store the \emph{marked} vertices. Since there can be only $h^{1/2 + \varepsilon/2}$ such vertices, we need $O(h^{1/2 + \varepsilon/2}\log n)$ space. It is easy to see that for every $\varepsilon' > 0$, there exists an $\varepsilon > 0$ such that $O(h^{1/2 + \varepsilon/2}\log n) = O(h^{1/2+\varepsilon'})$
Note that, for the input auxiliary graph $H$, our algorithm recurses on vertex induced subgraphs whose size is $O(h^{1-\varepsilon})$. Hence, the depth of recursion is bounded by a constant. This increases the space required by at most a constant factor. Since the number of vertices in the initial auxiliary graph was itself $O(n^{1/2 + \varepsilon/2})$, we get the desired space bound.

From the correctness and the space-complexity analysis of the algorithm, we conclude the proof of Theorem~\ref{th:penny}

\section{Chordal Graphs}\label{chordal}

In this section, we study the \textsc{Reachibility} problem on directed chordal graphs and design a space-efficient algorithm. A graph is said to be chordal if every cycle of length at least four has a \emph{chord}, which is an edge joining two vertices that are not adjacent on the cycle.
A \emph{directed} chordal graph is a graph whose underlying undirected graph is chordal.
See~\cite{dirac, fulkerson, rose1976algorithmic} for the fundamental theory on chordal graphs. We adopt the algorithm of Gilbert~et al.~\cite{gilbert1984separator} for \emph{Vertex Separator} and analyze to obtain the desired space-bound. In~\cite{DBLP:conf/isaac/JainT19}, it was noted that \textsc{Reachibility} on chordal graphs could be solved in a space-efficient manner, however, without any formal explanation. Here we provide a detailed analysis of this claim.

Let $G=(V,E)$ be a directed chordal graph.
For completeness, we will state some of the known results that are relevant to us. Chordal graphs are also known as \emph{triangulated graphs, monotone transitive graphs, rigid circuit graphs, perfect elimination graphs} in the literature.

\paragraph*{\textbf{Finding Separator.}}
Gilbert et al.~\cite{gilbert1984separator} presented $O(mn)$-time algorithm for finding \emph{Vertex Separator} on chordal graphs. Furthermore, a better algorithm of time complexity $O(m)$ is also shown. We design a space-efficient algorithm for finding a vertex separator that uses $O(m^{1/2}\log n)$ space and polynomial time.
We adopt the $O(mn)$-time algorithm from~\cite{gilbert1984separator} and analyze it to provide the desired space-bound. In order to do that, we allow time complexity that is much larger than $O(mn)$ but remains polynomial in $n$. We proceed with the following Lemma: A similar property is somewhat implicitly used by~\cite{gilbert1984separator}.

\begin{lemma}
\label{lem:chordalmain}
Let $G$ be a chordal graph and let $C$ be a clique. Let $A$ be the largest component in $G \setminus C$. Then, either of the following two statements is true:
\begin{itemize}
\item There exists a vertex $v$ in $C$ which is not adjacent to any vertex in $A$.
\item There exists a vertex $u$ in $A$ which is adjacent to every vertex in $C$.
\end{itemize}
\end{lemma}

We need several definitions and structural properties of the Chordal graphs in order to prove this Lemma. However, some of the properties are well-known results for chordal graphs.
We proceed with the following definitions.

\begin{definition}
Let $G$ be a graph and $v$ be a vertex of $G$. The deficiency of $v$, denoted by $D(v)$ is defined as follows: $D(v) = \{\{u_1, u_2\} \mid \{u_1, v\} \in E(G), \{u_2, v\} \in E(G) \textsf{ and }\{u_1, u_2\}\notin E(G)\}$
\end{definition}

\begin{definition}
Let $G$ be a graph and $v$ be a vertex of $G$. We define the graph $G_v$ as follows: $G_v = (V(G) \setminus \{v\}, E(G[V(G)\setminus\{v\}]) \cup D(v))$. We say that the graph $G_v$ is formed by eliminating $v$ from $G$.
\end{definition}

\begin{definition}
When a sequence of vertices is eliminated from a graph, the edges in the deficiencies that are added are called \emph{fill-in} edges. A \emph{simplicial} vertex of a graph is a vertex that has a deficiency of $0$.
\end{definition}

We state the following known facts about the chordal graphs. The following Lemmas are due to \cite{dirac, fulkerson, rose1972graph, rose1970triangulated, rose1976algorithmic}. For the sake of completion, we state them here in the form that we will be using and provide their proofs.

\begin{lemma}\label{lem:sepcliqueab}
Let $G$ be a chordal graph and $a$ and $b$ be two vertices in $V(G)$. Let $S$ be a set of vertices of $G$ such that: 1) $a$ and $b$ are in different components of the graph $G\setminus S$;
2) there exists no proper subset $S'\subseteq S$ such that vertices $a,b$ are in different components of the graph $G\setminus S'$.
If these conditions hold, then the set $S$ forms a clique in $G$.
\end{lemma}

\begin{proof}
Let $C_a$ and $C_b$ be the components in $G\setminus S$ containing $a$ and $b$, respectively. Note that, each vertex $s\in S$ is adjacent to some vertices in $C_a$, and some vertices in $C_b$. Consider two vertices $x,y\in S$. Let $\mathcal{P}_{a}$ be the shortest path between $x$ and $y$ in $G[C_a \cup \{x, y\}]$. Similarly, let $\mathcal{P}_b$ be the shortest path between $x$ and $y$ in $G[C_b \cup \{x, y\}]$. The paths $\mathcal{P}_a$ and $\mathcal{P}_b$ together forms a cycle.
Hence, there must be an edge between $x$ and $y$, since it is the only \emph{chord} that is possible. This argument holds for any pair of vertices in $S$. This completes the proof.
\end{proof}

\begin{lemma}
\label{lem:chordalauxiliary}
Let $G$ be a chordal graph and let $C$ be any clique of $G$. Then, either $G$ is a complete graph or there exists a vertex $v\in G \setminus C$ that is simplicial.
\end{lemma}

\begin{proof}
We prove this by induction on $\lvert G \rvert$. The base case $\lvert G \rvert = 1$ is trivial.
Let us assume that the Lemma holds for all $G$ such that $\lvert G \rvert \leq k$, for some $k\in [n]$.
Now, $G$ be a graph such that $\lvert G \rvert = k+1$.
If $G$ is not a complete graph then,
let $a,b$ be two vertices of $G$, that are not adjacent.
Due to Lemma~\ref{lem:sepcliqueab}, there exists a set $S$ that separates $a$ and $b$.
Let $C_a$, $C_b$, be the corresponding components $G$ containing $a$ and $b$, respectively.
Note that, the vertices in $C\setminus S$ should be in one component. W.l.o.g., assume that they belong to $C_a$.
Consider the graph $G_b = G[S \cup C_b]$. We have $\lvert G_b \rvert \leq k$. Hence, by induction, either $G_b$ is a clique or there is a vertex $u \notin S$, that is \emph{simplical} in $G_b$.
In either of the cases, there exists a vertex $u \notin S$ that is simplicial in $G_b$ since $G_b$ must contain at least one vertex that is not in $S$. Note, $u$ is a simplicial vertex in $G$ since $u$ is not adjacent to a vertex in any component other than $C_b$. This completes the proof.
\end{proof}

\begin{corollary}
Let $G$ be a chordal graph and $C$ be a clique. Let $A$ be the largest component in $G \setminus C$.
If $B$ is a non-empty subset of $A$,
then $B$ contains a vertex whose neighbours in $G[B\cup C]$ forms a clique.
\end{corollary}

\begin{lemma}[\cite{gilbert1984separator}]
\label{lem:chordalaux2}
Let $a_0, a_1, \ldots, a_k$ be an elimination ordering for a graph $G$. Let $v$ and $w$ be nonadjacent vertices of $G$. Then $\{v, w\}$ is a fill-in edge if and only if there is a path from $v$ to $w$ consisting of vertices that are eliminated earlier than both $v$ and $w$.
\end{lemma}

Now, we have the ingredients to prove the main Lemma~\ref{lem:chordalmain}.

\begin{proof}
Let $G$ be a chordal graph, and let $C$ be a clique. Let $A$ be the largest component in $G \setminus C$. Let us assume that each vertex $v\in C$ is adjacent to at least one vertex in $A$. We will show that under this assumption, there exists a vertex $u$ in $A$ that is adjacent to each vertex in $C$, thereby proving the Lemma.

It is known that a vertex-induced subgraph of a chordal graph is also a chordal graph. Hence, the subgraph $G_0 := G[A\cup C]$ is a chordal graph. Due to Lemma~\ref{lem:chordalauxiliary}, we know that there is a vertex $u_0\in G_0\setminus C$,
that is simplicial, whose neighbours form a clique. Let $G_1 := G[(A\cup C)\setminus \{u_0\}]$. Similarly, let $u_1$ be a vertex in $G_1\setminus C$ that is simplicial.
In general, let $u_i$ be a vertex in $G_i\setminus C$ whose neigbhours forms a clique, where $G_i$ is defined as $G[(A\cup C)\setminus \{u_0, u_1, \ldots, u_{i-1}\}]$.
Let $k$ be an integer such that
$G_{k+1} = G[C]$. We claim that $u_k$ is adjacent to every vertex of $C$.

Consider a vertex $x$ in $C$.
Since $A$ is connected and by our assumption, $x$ is adjacent to a vertex of $A$, there is a path from $x$ to $u_k$ in $G[A\cup C]$ that uses only vertices of $A$. Lemma~\ref{lem:chordalaux2} says that if $\{x,u_k\}$ is not an edge of
$G[A \cup C]$, then it is a fill-in edge. But a perfect elimination ordering has no fill-in, so $x$ is adjacent to $u_k$ in $G[A \cup C]$. Thus $u_k$ is adjacent to every vertex of $C$.
This completes the proof.
\end{proof}

Given Lemma ~\ref{lem:chordalmain}, the algorithm for finding a separator in a chordal graph is relatively straightforward, as also observed by~\cite{gilbert1984separator}.
The algorithm is described in Procedure~\ref{procedure}.

\begin{algorithm}
\SetKwFor{For}{for}{do}{endfor}
\SetKwFor{ForEach}{for each}{do}{endfor}
\SetKwIF{If}{ElseIf}{Else}{if}{then}{else if}{else}{endif}
\KwIn{A chordal graph $G$}
\KwOut{A separator $S$ of $G$ of size $\sqrt m$}
$S \gets \emptyset$\;
\While{there exists a component $A$ of $G \setminus S$ that has weight more than $n/2$}{
\While{there exists a vertex $x$ of $S$ that is adjacent to no vertex of $A$}{
$S \gets S \setminus \{x\}$\;
}
$v \gets$ a vertex of $A$ adjacent to every vertex of $S$\;
$S \gets S \cup \{v\}$\;
}
Return $S$\;
\caption{Vertex Separator in Chordal Graphs}
\label{procedure}
\end{algorithm}

The correctness of the algorithm directly comes from Lemma~\ref{lem:chordalmain}. We show that this algorithm can be implemented in $O(m^{1/2}\log n)$ space.

\subsection{Space-Complexity} To see that the above algorithm can be implemented in $O(m^{1/2}\log n)$-space and polynomial time, we first recall that the separator $S$ forms a clique in the graph. Since the size of a clique is upper-bounded by $m^{1/2}$, it will require at most $O(m^{1/2}\log n)$ space to store the clique.

Now, in order to implement the algorithm, we need to calculate the weights of each of the connected components of $G\setminus S$.
However, we cannot afford to store these components explicitly,
since the number of vertices present in these components could be large. Instead, we identify a component with the lowest-index vertex present in it. We call the lowest-index vertex a \emph{marker} of the component. In order to check that whether a vertex $v$ is a \emph{marker}, we run Reingold's undirected \textsc{Reachability} algorithm to see if it is not connected with any vertex $u$ in $G\setminus S$ such that the index of $u$ is lower than $v$. Also, to know the size of the connected component for a marker vertex $v$, we use Reingold's algorithm~\cite{reingold2008undirected} to count the number of vertices that are connected to it. Therefore it is possible to count the weight of any component of $G\setminus S$ in the desired space-bound. We conclude the following theorem.

\begin{theorem}\label{th-chordal}
Given a chordal graph $G=(V, E)$, there exists an algorithm that computes a $\sqrt{m}$ separator in polynomial time by using $O(m^{1/2}\log n)$-space.
\end{theorem}

From Theorem~\ref{th-chordal} and Theorem~\ref{treewidth}, we have the following corollary.

\begin{corollary}
There exists an algorithm that solves the \textsc{Reachability} problem for chordal graphs by using $O(m^{1/2}\log n)$-space and polynomial time.
\end{corollary}

\section{Conclusion}

We studied \textsc{Reachability} problem on three important graph families and obtained space-efficient algorithms for each of these classes. An interesting open problem is whether one can obtain a space-efficient algorithm for intersection graphs of Jordan regions when the embedding of the graph is not provided in the input. Another significant open problem is to study the \textsc{Reachability} for unit disk intersection graphs and obtain a space-efficient algorithm. The method that we used for unit contact disk graphs does not generalize to unit disk intersection graphs. In the latter case, there can be arbitrarily large directed cliques, and it is not possible to obtain auxiliary graphs while preserving reachability information between every pair of vertices. However, we believe that our method can be used to solve to \textsc{Reachability} for other classes of geometric contact graphs.

\bibliography{ref}

\end{document}